\documentclass[a4paper,11pt]{article}

\usepackage{fullpage}
\usepackage{graphicx}
\usepackage{amsmath}
\usepackage{amsthm}
\usepackage{enumerate}
\usepackage[tight]{subfigure}
\usepackage{colortbl}

\newtheorem{theorem}{Theorem}[section]
\newtheorem{lemma}[theorem]{Lemma}

\title{Improved Distance Queries in Planar Graphs}
\author{Yahav Nussbaum\thanks{The Blavatnik School of Computer Science, Tel Aviv University, 69978 Tel Aviv, Israel. Email address: \texttt{yahav.nussbaum@cs.tau.ac.il}}}
\date{}

\begin{document}

\maketitle

\begin{abstract}

There are several known data structures that answer distance queries between two arbitrary vertices in a planar graph. The tradeoff is among preprocessing time, storage space and query time. In this paper we present three data structures that answer such queries, each with its own advantage over previous data structures. The first one improves the query time of data structures of linear space. The second improves the preprocessing time of data structures with a space bound of $O(n^{4/3})$ or higher while matching the best known query time. The third data structure improves the query time for a similar range of space bounds, at the expense of a longer preprocessing time. The techniques that we use include modifying the parameters of planar graph decompositions, combining the different advantages of existing data structures, and using the Monge property for finding minimum elements of matrices.

\end{abstract}

\section{Introduction}

There are several known data structures that answer distance queries in planar graphs. We survey them below. All of these data structures use the following basic idea. They split the graph into \emph{pieces}, where each piece is connected to the rest of the graph only through its \emph{boundary vertices}. Then, every path can go from one piece to another only through these boundary vertices. The different data structures find different efficient ways to store or compute the distance between two boundary vertices or between a boundary vertex and a non-boundary vertex.

Frederickson \cite{F87} gave the first data structures that can answer distance queries in planar graphs fast. He gave a data structure of linear size with $O(n \log n)$ preprocessing time that can find the shortest path tree rooted at any vertex in $O(n)$ time, where $n$ is the number of vertices in the graph. This leads also to an $O(n^2)$ solution to the all-pairs shortest-paths problem, and implies a distance query data structure of size $O(n^2)$ with $O(1)$ query time. Feuerstein and Marchetti-Spaccamela \cite{FM93} modified the data structure of \cite{F87} and showed how to decrease the time of a distance query by increasing the preprocessing time. They do not provide an analysis of their data structure in terms of preprocessing time, storage space, and query time, but they show the total running time of $k$ queries, which is $O(nk + n \log n)$, $O(n^{4/3}k^{1/3})$, $O(n^{5/3})$, $O(n\sqrt{k})$ for $k \leq \sqrt{n}$, $\sqrt{n} \leq k \leq n$, $n \leq k \leq n^{4/3}$, $n^{4/3} \leq k \leq n^2$, respectively. This solution actually consists of three different data structures for the three cases $k \leq \sqrt{n}$, $\sqrt{n} < k \leq n$ and $n < k$, where the data structure for the first case is the one of~\cite{F87}.

Henzinger, Klein, Rao and Subramanian \cite{HKRS97} gave an $O(n)$ algorithm for the single-source shortest path problem. This implies a trivial distance query data structure, which uses the algorithm, and takes $O(n)$ space and query time.

Djidjev \cite{D96} gave three data structures. We will use the specific section number in \cite{D96} -- \S3, \S4, or \S5, to refer to each one of them. The first one \cite[(\S3)]{D96} works for $S \in [n^{3/2}, n^2]$ and has size $O(S)$, $O(S)$ preprocessing time, and $O(n^2/S)$ query time. The same data structure was also presented by Arikati et al.~\cite{ACC+96}. This data structure is similar to the two data structures of \cite{FM93}, but takes advantage of the algorithm of \cite{HKRS97} to get a better preprocessing time. The second \cite[(\S4)]{D96} works for $S \in [n, n^{3/2}]$ and has size $O(S)$, $O(n\sqrt{S})$ preprocessing time, and $O(n^2/S)$ query time. The third data structure \cite[(\S5)]{D96} works for $S \in [n^{4/3}, n^2]$, has size $O(S)$, $O(n\sqrt{S})$ preprocessing time, and $O(n \log(n/\sqrt{S}) / \sqrt{S})$ query time.\footnote{Djidjev \cite[(\S5)]{D96} presents the result for $S \in [n^{4/3}, n^{3/2}]$ with $O(n \log (n) / \sqrt{S})$ query time, however the same data structure works for a larger range of $S$, and the running time is actually $O(n \log(n/\sqrt{S}) / \sqrt{S})$.\label{ftnt:D96}} Chen and Xu \cite{CX00} presented a data structure with the same time and space bounds.\footnote{Chen and Xu \cite{CX00} do not bound the time and space of the data structure in terms of $S$, the bounds here are derived by setting $r = n^2/S$ in the bounds that appear below Lemma 28 (page 477) of \cite{CX00}; the bounds stated in \cite{CX00} depend on the minimum number $p$ of faces required to cover all vertices of the graph ($p$ is called \emph{face-on-vertex covering}), these bounds are obtained using hammock decomposition \cite{F95}, which can be applied to any planar distance data structure.\label{ftnt:CX00}}

Fakcharoenphol and Rao \cite{FR06} gave a data structure with $O(n \log n)$ space and $O(n \log^3 n)$ preprocessing time and $O(\sqrt{n}\log^2 n)$ query time. Klein \cite{K05} improved the preprocessing time of the data structure to $O(n \log^2 n)$.

Cabello \cite{C06} presented a data structure that uses $O(S)$ space and can be constructed in $O(S)$ time for $S \in [n^{4/3}\log^{1/3} n, n^2]$, and answers a query in $O(n \log^{3/2} (n)/\sqrt{S})$ time. This data structure can answer $k$ queries in a total of $O(n^{4/3}\log^{1/3}n+k^{2/3}n^{2/3}\log n)$ time. If the queries are known in advance, the algorithm of \cite{C06} avoids storing the entire structure, and uses only $O(n + k)$ space.

Other data structures exist for outerplanar graphs (some use the fact that these graphs have treewidth $2$) and planar graphs with small face-on-vertex covering \cite{B93,CZ00,DPZ91,DPZ95,F95,F96}, for planar graphs with bounded edge lengths \cite{E99,KK06}, and for approximate distance queries in planar graphs \cite{ACC+96,C95,K02,K05,T04}.

In this paper we present three new data structures for the problem:

\begin{itemize}
	\item Section \ref{sec:A}: A data structure with $O(n)$ space, $O(n \log n)$ preprocessing time, and $O(n^{1/2+\varepsilon})$ query time, for a constant $\varepsilon > 0$. This data structure has the best known query time achievable with linear space. This also improves the total running time for answering $k$ distance queries when $k$ is $O(n^{1/2-\varepsilon})$ and $\omega(\log n)$, and if we limit ourselves to data structures with $O(n + k)$ space then the upper bound on the range of $k$ grows to $O(n^{5/6-\varepsilon})$. The data structure is based on the data structure of \cite{FR06}, by replacing the \emph{recursive decomposition} of \cite{FR06} with recursive \emph{$r$-decomposition} as in \cite{F87}. As the data structure of \cite{FR06}, our data structure also generalizes to graphs embedded in a surface of bounded genus.
	\item Section \ref{sec:B}: For $S \in [n^{4/3},n^2]$, a data structure with $O(S)$ space, $O(S \log n)$ preprocessing time, and $O(n \log(n/\sqrt{S}) / \sqrt{S})$ query time. This data structure matches the storage space and query time of \cite[(\S5)]{CX00,D96}, which is the best query time for this range of storage space, with a better preprocessing time for $S = o(n^2/\log^2 n)$. The data structure is obtained by combining a preprocessing algorithm similar to \cite{C06} with a data structure similar to \cite[(\S5)]{D96}.
	\item Section \ref{sec:C}: For $S \in [n^{4/3},n^2]$, a data structure with $O(S)$ space, requiring $O((S^{3/2} / \sqrt{n}) \log^2 n)$ preprocessing time, and $O(n / \sqrt{S})$ query time. This data structure has the fastest query time for the same range of storage space as the previous data structure, but with a longer preprocessing time. The fast query time is obtained using an efficient minimum search in a Monge matrix.
\end{itemize}

The different data structures are summarized in Table \ref{table:sum}.

\begin{table}[t]
	\centering
	\small
	\begin{tabular}{|c|c|c|c|}
		\hline
		Reference & Storage space $O(S)$ & Query time & Preprocessing time\\
		\hline
		\cite{F87,HKRS97} & $S = n^2$ & $O(1)$ & $O(n^2)$\\
		\cite[(\S3)]{ACC+96,D96} & $S \in [n^{3/2}, n^2]$ & $O(n^2/S)$ & $O(S)$\\
		\cite{C06} & $S \in [n^{4/3}\log^{1/3}n,n^2]$ & $O(n \log^{3/2}(n) /\sqrt{S})$ & $O(S)$\\
		\cite[(\S5)]{CX00,D96} & $S \in [n^{4/3},n^2]$ & $O(n \log(n/\sqrt{S}) / \sqrt{S})$ & $O(n\sqrt{S})$\\
		\rowcolor[gray]{0.9}
		Section \ref{sec:B} & $S \in [n^{4/3},n^2]$ & $O(n \log(n/\sqrt{S}) / \sqrt{S})$ & $O(S\log n)$\\
		\rowcolor[gray]{0.9}
		Section \ref{sec:C} & $S \in [n^{4/3},n^2]$ & $O(n / \sqrt{S})$ & $O((S^{3/2} / \sqrt{n}) \log^2 n)$\\
		\cite[(\S4)]{D96} & $S \in [n, n^{3/2}]$ & $O(n^2/S)$ & $O(n\sqrt{S})$\\
		\cite{FR06}+\cite{K05} & $S = n \log n$ & $O(\sqrt{n}\log^2 n)$ & $O(n\log^2 n)$\\
		\cite{F87} & $S = n$ & $O(n)$ & $O(n\log n)$\\
		\cite{HKRS97} & $S = n$ & $O(n)$ & ---\\
		\rowcolor[gray]{0.9}
		Section \ref{sec:A} & $S = n$ & $O(n^{1/2+\varepsilon})$ & $O(n \log n)$\\
		\hline
	\end{tabular}
	\caption{Comparison of distance query data structures for planar graphs. Time bounds are expressed as a function of the storage space. The data structures are ordered by decreasing storage space and then by decreasing query time.}
	\label{table:sum}
\end{table}

\section{Preliminaries} \label{sec:pre}

We consider a directed simple planar graph $G(V,E)$. We let $n = |V(G)|$, and by Euler's formula $|E(G)| = O(n)$. We assume that $G$ is given with a fixed planar embedding, in other words it is a \emph{plane graph}. Without loss of generality we assume that $G$ is a triangulated, bounded degree graph; this assumption is required by algorithms that we use and are described in Sect.~\ref{sec:decomp} and \ref{sec:dense}. We assume that $G$ is connected, since we can handle each connected component separately.

Every edge in $E(G)$ has a non-negative \emph{length}. The length of a \emph{path} is the sum of lengths of all of its edges. The \emph{distance} from a vertex $u$ to a vertex $v$ is the minimum length of a path from $u$ to $v$. With additional $O(n\log^2n/\log\log n)$ preprocessing time we can allow negative edge lengths as well, see \cite{FR06,KMW10,MW10} for details.

Let $F, H$ be subgraphs of $G$. We write $d_H(u,v)$ to denote the distance from $u$ to $v$ in $H$. The graph $F \cap H$ is the subgraph induced by $E(F) \cap E(H)$. For short we denote $|H| = |V(H)|$.

\subsection{Decomposition} \label{sec:decomp}

A \emph{decomposition} of a planar graph $G$ is a set of subgraphs of $G$, such that each edge is in exactly one subgraph and each vertex of $G$ is in at least one subgraph. Each of the subgraphs which define the decomposition is called a \emph{piece}.

A vertex $v$ is a \emph{boundary vertex} of a piece $B$, if $v \in V(B)$ and $v$ is incident to some edge not in $E(B)$. The set of all boundary vertices of $B$ is the \emph{boundary} of $B$, denoted by $\partial B$. A \emph{hole} is a face of $B$ (including the external face) that is not a face of $G$. For a hole $H$ we denote by $H$ also the subgraph of $G$ inside $H$. A \emph{boundary walk} of $B$ is a facial walk of $B$ around a hole $H$. For a piece $B$ with hole $H$ we denote $\partial B[H] = \partial B \cap V(H)$. A vertex of $V(B) \setminus \partial B$ is an \emph{internal vertex} of $B$.

All distance query data structures mentioned in the introduction decompose the planar graph. They take advantage of the fact that a path can go from one piece to another only through boundary vertices.

A \emph{recursive decomposition} \cite{FR06} is obtained by starting with $G$ itself being the only piece in level 0 of the decomposition. At each level, we split each piece $B$ with $|B|$ vertices and $|\partial B|$ boundary vertices that has more than one edge into two pieces, each with at most $2|B|/3$ vertices and at most $2|\partial B|/3 + O(\sqrt{|B|})$ boundary vertices. We require that the boundary vertices of a piece $B$ are also boundary vertices of the subpieces of $B$. Each piece $B$ in the decomposition has $O(\sqrt{|B|})$ boundary vertices.

An \emph{$r$-decomposition} \cite{F87} is a decomposition of the graph into $O(n/r)$ pieces, each of size at most $r$ with $O(\sqrt{r})$ boundary vertices.

Fakcharoenphol and Rao \cite{FR06} showed how to find a recursive decomposition of $G$, such that each piece is connected and has at most a constant number of holes. They use these two properties for their distance algorithm. The construction of the decomposition takes $O(n \log n)$ time using $O(n \log n)$ space, and is done by recursively applying the separator algorithm of Miller \cite{M86}. Frederickson \cite{F87} showed how to find an $r$-decomposition in $O(n \log n)$ time and $O(n)$ space by recursively applying the separator algorithm of Lipton and Tarjan \cite{LT79}. Thus, an $r$-decomposition is a limited type of recursive decomposition where we stop the recursion earlier (when we get to pieces of size $r$), and do not store all the levels of the recursion (we store only the leaves). Cabello \cite{C06} combined the two constructions of \cite{FR06,F87} (using \cite{M86} instead of \cite{LT79}) and constructed an $r$-decomposition with the properties that the number of holes per piece is bounded by a constant, and that each piece is connected.

In Sect.~\ref{sec:A} we use a combination of recursive decomposition and $r$-decomposi\-tion -- we decompose the graph recursively, but we decompose each piece into $O(n/r)$ pieces instead of 2. In Sect.~\ref{sec:B} we use $r$-decomposition. In Sect.~\ref{sec:C} we use $r$-decomposition as well, there we take advantage of the fact that the construction of an $r$-decomposition is the same as of a recursive decomposition, which was stopped earlier.

\subsection{The Dense Distance Graph} \label{sec:dense}

Fakcharoenphol and Rao \cite{FR06} define the \emph{dense distance graph} of a recursive decomposition. For each piece $B$ in the recursive decomposition they add a piece to the dense distance graph that contains the vertices of $\partial B$ and for every $u, v \in \partial B$ an edge from $u$ to $v$ whose length is $d_B(u, v)$. The multiple-source shortest paths algorithm of Klein \cite{K05} finds $k$ distances where the sources of all of them are on the same face in $O((k + n) \log n)$ time. Therefore, using \cite{K05} it takes $O(|B| \log |B|)$ time to find the part of the dense distance graph that corresponds to a piece $B$ (recall that $|\partial B| = O(\sqrt{|B|})$ and $B$ has a constant number of holes). It thus takes $O(n \log^2 n)$ time to construct the dense distance graph over all pieces of the recursive decomposition.

Every single edge defines a piece in the base of the recursion, so it is clear that the distance from $u$ to $v$ in the dense distance graph is the same as the distance between these two vertices in the original graph. Fakcharoenphol and Rao noticed that in order to find the distance from $u$ to $v$ we do not have to search the entire dense distance graph, but that it suffices to consider only edges that correspond to shortest paths between boundary vertices in a limited number of pieces. The pieces are these containing either $u$ or $v$, and their siblings in the recursive decomposition. There are $O(\log n)$ such pieces with a total of $O(\sqrt{n})$ boundary vertices. Fakcharoenphol and Rao gave an implementation of Dijkstra's algorithm that runs over a subgraph of the dense distance graph with $q$ vertices, defined by a partial set of the pieces in the recursive decomposition, in $O(q \log^2 n)$ time. This gives the $O(\sqrt{n} \log^2 n)$ query time of their data structure.

We use dense distance graphs in two of our data structures (Sect.~\ref{sec:A} and \ref{sec:C}). In both cases it is on a variant of recursive decomposition, as discussed above.

\subsection{The Monge Property} \label{sec:Monge}

A $p \times q$ matrix $M$ satisfies the \emph{Monge property} if for every two rows $i \leq k$ and two columns $j \leq \ell$, $M$ satisfies $M_{ij} + M_{k\ell} \leq M_{i\ell} + M_{kj}$. We can find the minimum element of $M$ by transposing, negating and reversing $M$, and using the SMAWK algorithm \cite{AKM+87} for row-maxima on the resulting totally monotone matrix. If we do not store $M$ explicitly, but are able to retrieve each entry in $O(1)$ time this takes $O(p+q)$ time. Note that if we add a constant to an entire row or to an entire column of a matrix with the Monge property, then the property remains.

Consider two disjoint sets $X$ and $Y$ of consecutive boundary vertices on a boundary walk of some piece $B$. Rank the vertices of $X$ from $x_1$ to $x_{|X|}$ according to their order around the boundary walk, and rank the vertices of $Y$ from $y_1$ to $y_{|Y|}$ according to their order in the opposite direction around the boundary walk. For $i \leq k$ and $j \leq \ell$, the shortest path from $x_i$ to $y_\ell$ inside $B$ and the shortest path from $x_k$ to $y_j$ inside $B$ must cross each other. Let $w$ be a vertex common to both paths. Then, $d_B(x_i,y_j) + d_B(x_k,y_\ell) \leq d_B(x_i,w) +d_B(w,y_j) + d_B(x_k,w) + d_B(w,y_\ell) \leq d_B(x_i,y_\ell) + d_B(x_k,y_j)$ (see Fig.~\ref{fig:monge}). Therefore, the matrix $M$ such that $M_{ij} = d_B(x_i,y_j)$ has the Monge property. The Monge property was first used explicitly for distance queries in planar graphs by \cite{FR06}.

\begin{figure}
	\centering
	\scriptsize
	\begin{minipage}{0.4\linewidth}
		\centering
		\includegraphics[scale=0.5]{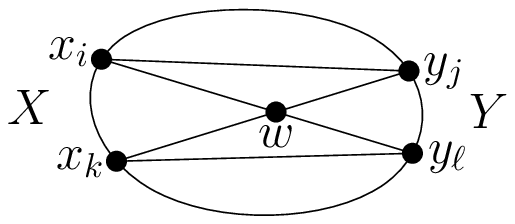}
		\caption{The distances from $X$ to $Y$ satisfy the Monge property}
		\label{fig:monge}
	\end{minipage}
	\qquad
	\begin{minipage}{0.4\linewidth}
		\begin{equation*}
			\centering
			\begin{bmatrix}
				1 & 2 & 3 & 4 \\
				& 5 & 6 & 7 \\
				& & 8 & 9 \\
				& & & 0 
			\end{bmatrix}
			\quad
			\begin{bmatrix}
				1 & & & \\
				2 & 3 & & \\
				4 & 5 & 6 & \\
				7 & 8 & 9 & 0 
			\end{bmatrix}
		\end{equation*}
		\caption{A falling staircase matrix (left) and an inverse falling staircase matrix (right)}
		\label{fig:stair}
	\end{minipage}
\end{figure}

A \emph{partial matrix} is a matrix that may have some blank entries. In a \emph{falling staircase matrix} the non-blank entries are consecutive in each row starting not before the first non-blank entry of the previous row and ending at the end of the row (see Fig.~\ref{fig:stair}), \emph{inverse falling staircase matrix} is defined similarly by exchanging the positions of the non-blanks and the blanks. Aggarwal and Klawe \cite{AK90} find the minimum of an (inverse) falling staircase matrix whose non-blank entries satisfy the Monge property in $O(q+p)$ time by filling the blanks with large enough values to create a Monge matrix.

In Sect.~\ref{sec:C} we use this tool for finding the minimum of two staircase matrices whose non-blank entries satisfy the Monge property.

\section{Linear-Space Data Structure} \label{sec:A}

In this section we present a data structure with linear space, almost linear preprocessing time, and query time faster than any previous data structure of linear space. We generalize the data structure of Fakcharoenphol and Rao \cite{FR06} by combining recursive decomposition of the graph with $r$-decomposition. This is similar to the way that Mozes and Wulff-Nilsen \cite{MW10} improved the shortest path algorithm of Klein, Mozes and Weimann \cite{KMW10}. Mozes and Sommer \cite{MS10} have independently obtained a similar result.

We find an $r$-decomposition of $G$ into $p = O(n / r)$ pieces, and then we recursively decompose each piece into $p$ subpieces, until we get to pieces with a single edge. The depth of the decomposition is $O(\log n/\log p)$ where at level $i$ we have $p^i$ pieces, each of size $O(n / p^i)$ and with $O(\sqrt{n / p^i})$ boundary vertices. Constructing this recursive decomposition takes $O(n \log n \frac{\log n}{\log p})$ time. An alternative way to describe this decomposition is to perform a recursive decomposition on $G$ while storing only levels $k \log p$ for $k = 0, \dots, \lfloor \log n / \log p \rfloor$ of the recursion tree and the leaves of the recursion (the pieces containing single edges).

We compute the dense distance graph for the recursive decomposition, in the same way as in the data structure of \cite{FR06}. That is, we compute the distance between every pair of boundary vertices in each piece. Using the algorithm of Klein \cite{K05} this takes $O(n \log n)$ time for each level, and a total of $O(n \log n \frac{\log n}{\log p})$ time. The size of dense distance graph over our recursive decomposition is $O(n \frac{\log n}{\log p})$.

When a distance query from $u$ to $v$ arrives, we use the Dijkstra implementation of \cite{FR06} to answer it. We run the algorithm on the subgraph of the dense distance graph that includes all the pieces that contain either $u$ or $v$, and the $p-1$ siblings in the recursive decomposition of each such piece. We require the sibling pieces because the shortest path can get out of a piece $B$ into a sibling of $B$ without getting out of any piece that contains $B$. Therefore, the number of boundary vertices involved in each distance query is $O(\sqrt{n}+\sum_{i=1}^{\log n / \log p} p\sqrt{n/p^i}) = O(p\sqrt{n})$. Hence the query time using the algorithm of \cite{FR06} is $O(p\sqrt{n}\log^2 n)$.

We conclude that
for a planar graph with $n$ vertices and any $p \in [2, n]$, we can construct in $O(n \log^2 n / \log p)$ time a data structure of size $O(n \log n / \log p)$ that computes the distance between any two vertices in $O(p\sqrt{n}\log^2 n)$ time.
If we set $p = 2$ we get exactly the data structure of \cite{FR06}. If we set $p = n^\delta$ for a constant $0 < \delta < \varepsilon$ we get:

\begin{theorem} \label{col:A}
  For a planar graph with $n$ vertices and any constant $\varepsilon > 0$, we can construct in $O(n \log n)$ time a data structure of size $O(n)$ that computes the distance between any two vertices in $O(n^{1/2+\varepsilon})$ time.
\end{theorem}

The total time for $k$ distance queries is $O(n \log n + kn ^ {1/2+\varepsilon})$ and the required space is $O(n + k)$. This improves the fastest time for $k$ distance queries for $k = O(n^{1/2-\varepsilon})$ and $k = \omega(\log n)$ simultaneously. Among data structures that require only $O(n + k)$ space, the upper bound on the range of $k$ is $O(n^{5/6-\varepsilon})$.

Fakcharoenphol and Rao \cite{FR06} noted that Smith suggested that their algorithm can be generalized to graphs of bounded genus. If a graph $G$ with bounded vertex degree is embedded in an orientable surface of genus $g$, then \cite{DV95,HM86} showed how to find a \emph{planarizing set} of $O(\sqrt{ng})$ edges whose removal from the graph makes the graph planar, in $O(n + g)$ time. We use the planarizing set for the first decomposition of the graph, and combine the Dijkstra implementation of \cite{FR06} with standard implementation using a heap for the topmost pieces in the recursion. We get that the bounds of Theorem \ref{col:A} apply also to graphs embedded in an orientable surface of a fixed genus.

\section{Improved Preprocessing Time for $S \in [n^{4/3}, n^2]$} \label{sec:B}

In this section we present a data structure that matches the space-query time tradeoff of the data structures of Djidjev \cite[(\S5)]{D96} and Chen and Xu \cite{CX00} with the preprocessing time of the data structure of Cabello \cite{C06}. Our data structure combines parts of the data structures of \cite[(\S5)]{D96} and of \cite{C06}.

First, we construct an $r$-decomposition of $G$ in $O(n \log n)$ time, for some parameter $r \in (0,n)$. For each piece $B$ our data structure has three parts:
\begin{enumerate}[(i)]
% Give labels
	\item The distances $d_G(u, v)$ and $d_G(v, u)$ for every $u \in \partial B$ and $v \in V(B)$.
	\item A data structure that reports $d_B(u, v)$ in $O(\sqrt{r})$ time for $u, v \in V(B)$.
	\item For each hole $H$ of $B$ we store $d_H(u, v)$ for every $u \in \partial B[H]$ and $v \in V(H)$ such that $v$ is a boundary vertex of some piece contained in $H$.
\end{enumerate}

% The references below can be replaced with \ref's
Part (i) is from the data structure of Cabello \cite{C06}. The construction of this part requires $O(n \log n + r^{3/2})$ time and $O(n + r^{3/2})$ space per piece \cite{C06}.
Part (ii) was used both by Djidjev \cite[(\S5)]{D96} and by Cabello \cite{C06}. This is the data structure of \cite[(\S3)]{ACC+96,D96} with $S = r^{3/2}$, its construction takes $O(r^{3/2})$ time and space per piece.
Part (iii) is from the data structure of \cite[(\S5)]{D96}, but we construct it more efficiently. We find the distances for this part using the multiple-source shortest paths algorithm of Klein \cite{K05} for every boundary walk. The required space per piece for part (iii) is $O(n)$ and the preprocessing time is $O(n \log n)$.

Since there are $O(n/r)$ pieces, each with a constant number of holes and $O(\sqrt{r})$ boundary vertices, constructing the three parts takes $O((n^2/r)\log n + n\sqrt{r})$ time and $O(n^2/r + n\sqrt{r})$ space.

Let $u, v$ be a query pair. We use the data structure of this section to find $d_G(u, v)$ in $O(\sqrt{r} \log r)$ time. If $u$ and $v$ are in the same piece then we find the distance from $u$ to $v$ using parts (i) and (ii) of the data structure with the query algorithm of \cite{C06} in $O(\sqrt{r})$ time (see details in Sect.~\ref{sec:Cq} below). If $u$ and $v$ are in different pieces then we find the distance using parts (i) and (iii) with the query algorithm of \cite[(\S5)]{D96} in $O(\sqrt{r} \log r)$ time (see details in Appendix \ref{apx:B}).

We conclude that
for a planar graph with $n$ vertices and any $r \in (0, n)$, we can construct in $O((n^2 / r) \log n + n \sqrt{r})$ time a data structure of size $O(n^2 / r + n\sqrt{r})$ that computes the distance between any two vertices in $O(\sqrt{r}\log r)$ time.
The sum $n^2 / r + n\sqrt{r}$ minimizes at $n^{4/3}$, and for $r = n^2 / S$ we get:
\begin{theorem}
  For a planar graph with $n$ vertices and $S \in [n^{4/3},n^2]$, we can construct in $O(S \log n)$ time a data structure of size $O(S)$ that computes the distance between any two vertices in $O(n \log(n/\sqrt{S}) / \sqrt{S})$ time.
\end{theorem}

\section{Improved Query Time for $S \in [n^{4/3}, n^2]$} \label{sec:C}

In this section we present a data structure with an improved query time, for the same range of space bounds as in the previous section. In return, the preprocessing time is higher. For this purpose we use minimum search in Monge matrices. While previous planar distance data structures have taken advantage of the Monge property before, this is the first to use fast minimum search in a Monge matrix with the SMAWK algorithm \cite{AKM+87}.

Again, we construct an $r$-decomposition of $G$. Assume that we want to find the distance from a vertex $u$ to a vertex $v$ that are in two different pieces. Let $B$ and $B'$ be the different pieces that contain $u$ and $v$ respectively, let $H$ and $H'$ be the holes of $B$ and $B'$ that contain $v$ and $u$ respectively, and let $X = \partial B[H]$ and $Y = \partial B'[H']$. Let $J = H \cap H'$ be the subgraph of $G$ contained both in $H$ and in $H'$. We assume without lost of generality that $J$ contains the infinite face. See Fig.~\ref{fig:subgraphJ}.

\begin{figure}
	\centering
	\includegraphics[scale=0.5]{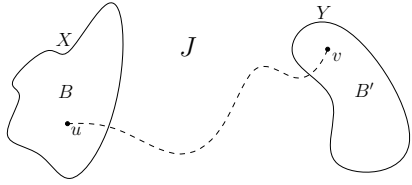}
	\caption{The setting in Sect.~\ref{sec:C}. The vertex $u$ is in the piece $B$, and $v$ is in $B'$. The subgraph $J$ is between $B$ and $B'$. Every path from $u$ to $v$ must contain a vertex from the boundary walks $X$ of $B$ and a vertex from the boundary walk $Y$ of $B'$.}
	\label{fig:subgraphJ}
\end{figure}

The shortest path from $u$ to $v$ must contain a vertex $x \in X$ and a vertex $y \in Y$ (it is possible that $x = y$). We assume that there is no internal vertex of $B$ or $B'$ between $x$ and $y$ in this path, since otherwise we can replace $x$ with a later vertex. Therefore, $d_G(u, v) = \min_{x \in X, y \in Y} \{d_G(u, x) + d_J(x, y) + d_G(y, v)\}$.

Our goal then is to find $x, y$ that minimize $d_G(u, x) + d_J(x, y) + d_G(y, v)$. For a particular order of $X$ and of $Y$, which we specify below, let $M$ be the matrix such that $M_{ij} = d_J(x_i, y_j)$, and $N$ be the matrix such that $N_{ij} = d_G(u, x_i) + d_J(x_i, y_j) + d_G(y_j, v)$. We show how to order the members of $X$ and $Y$ such that $M$ decomposes into two staircase matrices, each with the Monge property. Since $d_G(u, x_i)$ is fixed for a fixed $x_i$, and $d_G(y_j, v)$ is fixed for a fixed $y_j$, then $N$ also consists of two staircase matrices with the Monge property. Thus we can use the algorithm of Aggarwal and Klawe \cite{AK90} to find the minimum entry of $N$, which is the desired distance.

For every $x \in X$ we define the \emph{leftmost shortest path from $x$ to $Y$}, denoted by $L(x)$ as follows. We add to the embedding of $J$ a vertex $u'$ inside $B$ and connect it with an edge to $x$, and a vertex $v'$ inside $B'$ and connect it with edges from every vertex of $Y$ (recall the internal vertices of $B$ and $B'$ are not in $J$). We set the length of all new edges to be $0$. An edge $e = (w,z) \in E(J)$ is called \emph{tight} if the length of $e$ is equal to $d_J(u',z) - d_J(u',w)$, that is if $e$ is on some shortest path from $u'$ to $z$. We remove all non-tight edges from the graph, and perform a \emph{left-first search} from $u'$ until we find $v'$ (i.e. we perform a depth-first search from $u'$, and visit the edges outgoing from a specific vertex according to their left-to-right order, see also \cite{K05}). Let $L(x)$ be the path we obtain by removing the first and the last edges of the leftmost path we found from $u'$ to $v'$, and let $\ell(x) \in Y$ be the last vertex of $L(x)$. Note that $L(x)$ is a shortest path from $x$ to $\ell(x)$. The reason we added $u'$ is to decide between two paths that diverge at $x$ itself, the reason we added $v'$ is to decide between two paths such that one is a prefix of the other. Note that $L(x)$ may contain more than one vertex of $Y$. Moreover, even if $x \in X \cap Y$ it is not necessarily true that $\ell(x) = x$. See Fig.~\ref{fig:leftfs}.

\begin{figure}
	\centering
	\includegraphics[scale=0.75]{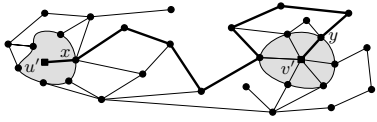}
	\caption{Finding $L(x)$. Only tight edges of $J$ are shown, $B$ and $B'$ are \emph{shaded}. The leftmost path from $u'$ to $v'$ is \emph{bold}, $y = \ell(x)$, $L(x)$ is its subpath between $x$ and $y$. Directions of the edges are not shown.}
	\label{fig:leftfs}
\end{figure}

Fix some arbitrary vertex of $X$ to be $x_1$ and rank the other vertices of $X$ in a clockwise order. Let $y_{|Y|} = \ell(x_1)$, and rank the vertices of $Y$ in a counterclockwise order. Let $P = L(x_1)$.

\begin{lemma} \label{lem:C}
	Let $x_i \in X$ and $y_j \in Y$. There is a shortest path $Q$ in $J$ from $x_i$ to $y_j$ such that either $Q$ does not cross $P$, or every prefix of $Q$ crosses $P$ from the left side of $P$ to its right side at most once more than it crosses $P$ from its right side to its left side.
\end{lemma}

\begin{proof}
Assume that every shortest path from $x_i$ to $y_j$ in $J$ crosses $P$, and let $Q$ be such a path. First, we can assume that at the first time that $Q$ emanates from $P$ it emanates from its right side -- if $Q$ emanates from the left side of $P$ at a vertex $w$, we can replace the suffix of $P$ that starts at $w$ with the suffix of $Q$ and get a path from $x_1$ to $y_j$ which is to the left of $P$, contradicting its definition as $L(x)$ (see Fig.~\ref{fig:QxP}(a)). Second, we may assume that if $Q$ meets $P$ at a vertex $w$, and then again at a vertex $w'$ such that $w'$ is after $w$ also in $P$, then the subpath of $Q$ between $w$ and $w'$ is the same subpath as in $P$, since otherwise we can replace the subpath of $Q$ between $w$ and $w'$ with the subpath of $P$ (see Fig.~\ref{fig:QxP}(b)). Last we may assume that in two consecutive times that $Q$ crosses $P$ it does so from different directions, since if the same direction is repeated twice, and by the previous observation the second crossing precedes the first crossing in their order in $P$, then $Q$ must cross itself (see Fig.~\ref{fig:QxP}(c)). From these three observations the lemma follows.
\end{proof}

\begin{figure}[t]
	\centering
	\subfigure[If $Q$ emanates left from $P$ at $w$ then we can replace the suffix of $P$ that starts at $w$ with a suffix of $Q$.]{\includegraphics[scale=0.5]{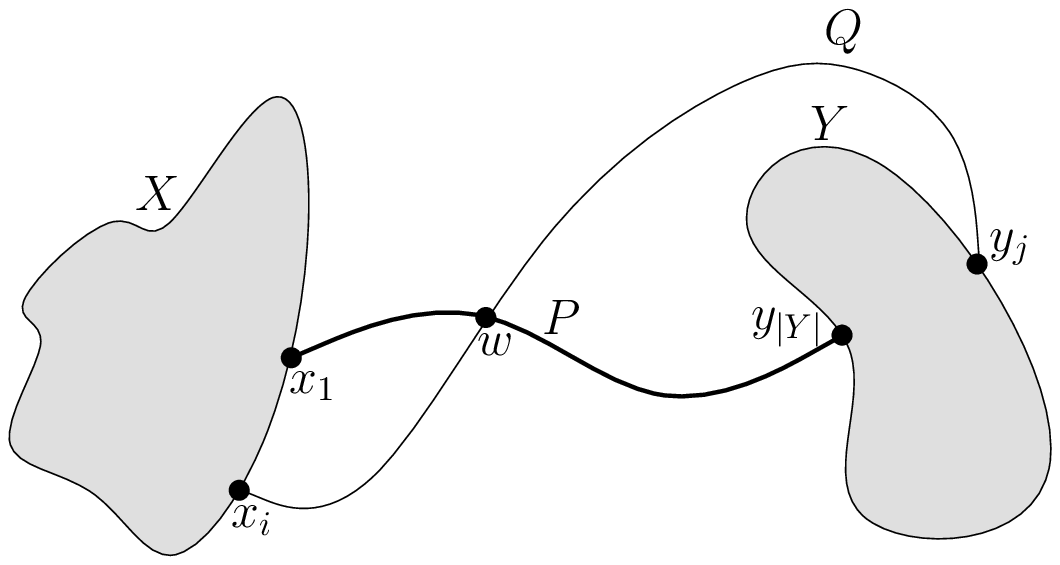}} \quad
	\subfigure[If $P$ and $Q$ contain $w$ and $w'$ in the same order then we can replace the subpath of $Q$ between $w$ and $w'$ with the subpath of $P$.]{\includegraphics[scale=0.5]{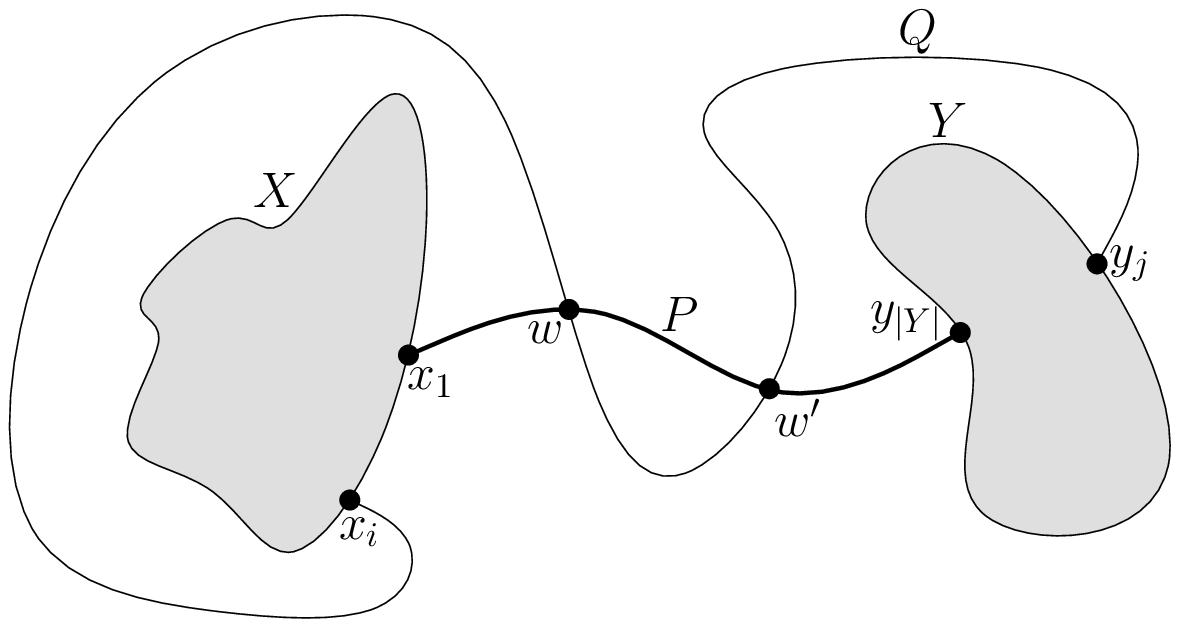}} \\
	\subfigure[If $Q$ crosses $P$ twice consecutively from the same side, then it must cross itself (only the first time is illustrated at $w'$).]{\includegraphics[scale=0.5]{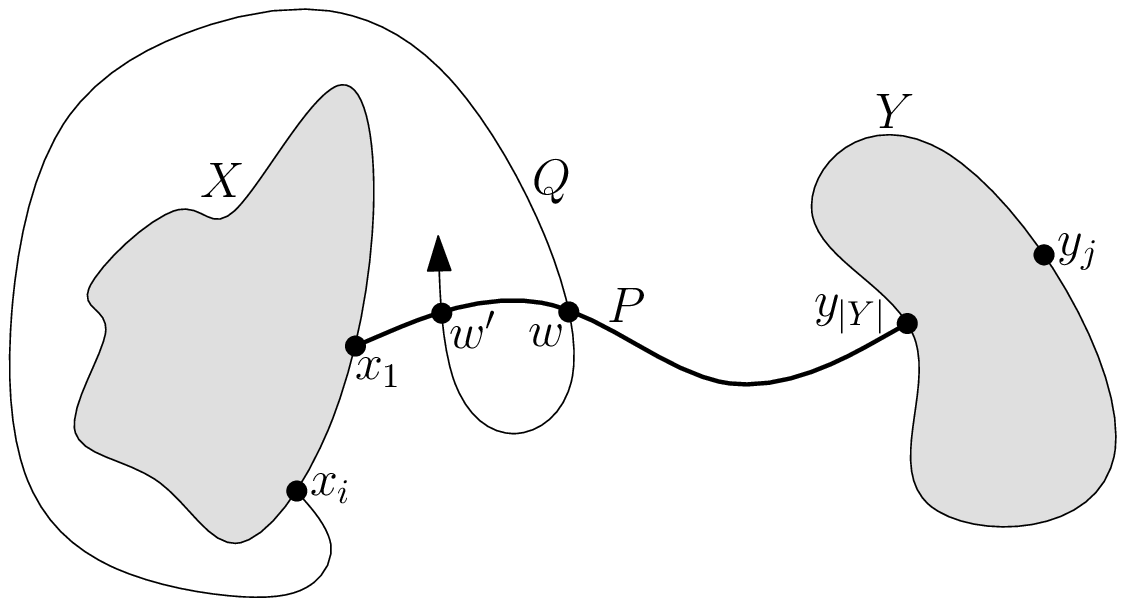}} \quad
	\subfigure[If a shortest path from $X$ to $Y$ does not cross $P$ (\emph{solid lines}), or always cross $P$ (\emph{dashed} lines), then $M$ has the Monge property (cf.\ Fig.~\ref{fig:monge}).]{\includegraphics[scale=0.5]{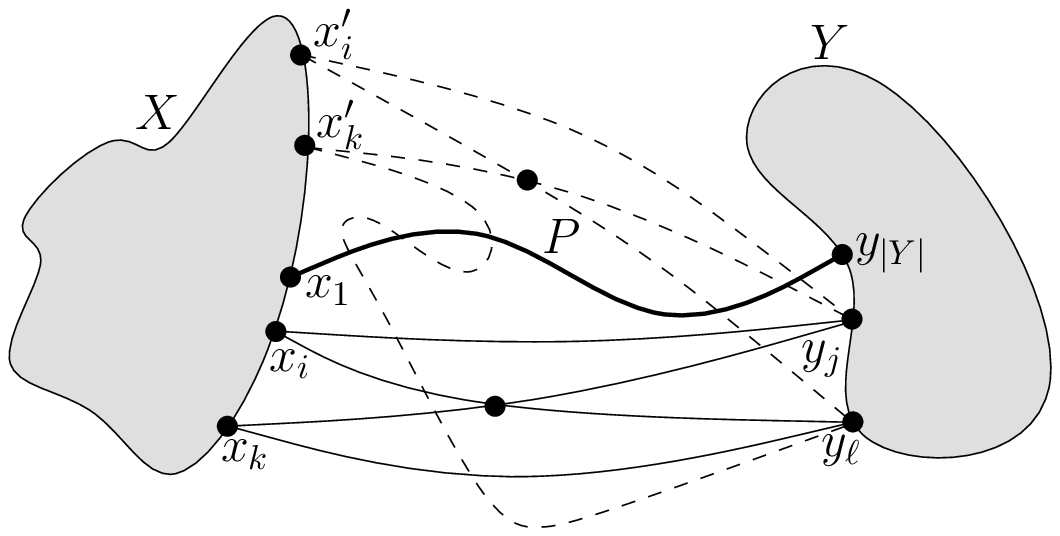}}
	\caption{Crossing of $Q$ and $P$.}
	\label{fig:QxP}
\end{figure}

Let $M'$ be the partial matrix of $M$ where $M'_{ij}$ is non-blank if there is a shortest path in $J$ from $x_i$ to $y_j$ that does not cross $P$, and let $M''$ be the partial matrix of $M$ where $M''_{ij}$ is non-blank if every shortest path from $x_i$ to $y_j$ crosses $P$. The partial matrix $M'$ has the Monge property, we get this by cutting open $J$ along $P$ and using the claim from Sect.~\ref{sec:Monge} (see $x_i, x_k$ in Fig.~\ref{fig:QxP}(d)). A similar argument (by taking two copies of $J$ open at $P$ and ``gluing'' the right side of one of them to the left side of the other) shows that $M''$ also has the Monge property (see $x'_i, x'_k$ in Fig.~\ref{fig:QxP}(d)). The non-blank entries of row $i$ in $M'$ are from $\ell(x_i)$ to $y_{|Y|} = \ell(x_1)$, and the rest of the row is in $M''$. The partial matrix $M'$ is a falling staircase matrix and $M''$ is an inverse falling staircase matrix, since $L(x_{i+1})$, the leftmost path from $x_{i+1}$ to $Y$, cannot cross the path $L(x_i)$ from its right side to its left (this is similar to the illustration in Fig.~\ref{fig:QxP}(a)). Let $N'$ and $N''$ be the corresponding staircase matrices partial to $N$ (with same blank entries as $M'$ and $M''$ respectively), both of them have the Monge property.

In order use the insights above, we take parts (i) and (ii) of the data structure of the previous section together with the following two new parts, where $X, Y$ and $J$ are as defined above (part (iii) is not necessary):
\begin{enumerate}[(i)]
% The fixed number below can be replaced with resume
\setcounter{enumi}{3}
% Give label
	\item For each two pieces $B$ and $B'$ we store $d_J(x, y)$ for every $x \in X$,  $y \in Y$.
	\item For each two pieces $B$ and $B'$, and every $x \in X$, we store $\ell(x)$.
\end{enumerate}

Notice that $d_J(x, y)$ and $\ell(x)$ depend on the specific pieces $B$ and $B'$. Since there are $n / r$ pieces, each piece has $\sqrt{r}$ boundary vertices, and each vertex is on the boundary of a constant number of pieces, the total space required for part (iv) is $O(n^2 / r)$ and for part (v) is $O(n^2 / r^{3/2})$. This does not increase the total space complexity of the data structure which remains $O(n^2 / r + n\sqrt{r})$.

\subsection{The Preprocessing Algorithm} \label{sec:Cp}

We have an $r$-decomposition of $G$ obtained by a recursive decomposition, where we decomposed each piece into two pieces using a separator and stopped the recursion at pieces of size $O(r)$, from which we took only the pieces that correspond to leaves of the recursion tree. Now we will take all the pieces of the entire recursive decomposition which defined the $r$-decomposition. We build a dense distance graph for $G$ based on this recursive decomposition using the algorithm of Fakcharoenphol and Rao \cite{FR06} with the improvement of Klein \cite{K05} in $O(n \log^2 n)$ time and $O(n \log n)$ space.

For two fixed pieces $B$ and $B'$ we use a subgraph of the dense distance graph to compute the distances from vertices of $X$ to vertices of $Y$ in $J$. We should choose carefully a set $S$ of pieces of the recursive decomposition to use, we must obey three rules -- we cannot take any piece of the $O(\log n)$ pieces that contain either $B$ or $B'$, we should cover all the paths from $X$ to $Y$, and the total number of boundary vertices of pieces in $S$ should not be too large.

We start with the entire graph $G$, which is the root of the recursive decomposition, as the single piece in $S$. As long as there is a piece $C$ in $S$ containing either $B$ or $B'$ in it ($C$ is an ancestor of $B$ or $B'$ in the recursive decomposition), we replace $C$ with both of its two children in the recursion tree. When we get to $B$ or $B'$, we remove them from $S$. At the end of this process, $S$ contains $O(\log n)$ pieces with $O(\sqrt{n})$ boundary vertices. All the vertices of $X$ and $Y$ are boundary vertices of pieces of $S$ (because otherwise such a vertex will be internal vertex of some member of $S$, which is an ancestor of $B$ or $B'$), and the paths between them in $J$ are covered by $S$ (since $S$ only misses internal vertices of $B$ and $B'$), as required. Denote the subgraph of the dense distance graph that include exactly the pieces of $S$ by $D$.

For $x \in X$, we compute part (iv), by computing $d_J(x, y)$ for every $y \in Y$ using the Dijkstra implementation of \cite{FR06} on $D$ in $O(\sqrt{n} \log^2 n)$ time.

We use the distances of vertices of $D$ from $x$ to find $\ell(x)$ for part (v) as well. We add to $D$ a vertex $u'$ inside $B$ and connect it to $x$, and a vertex $v'$ inside $B'$ and connect every vertex of $Y$ to it, as described before. Since we added an edge with length $0$ from $u'$ to $x$, for every vertex $z$ in $D$, $d_J(u',z) = d_J(x,z)$, so we can determine for each edge of $D$ whether it is tight in a constant time. Even though $D$ is not planar, we can define a cyclic order on the edges incident to each vertex using the embedding of $G$. The order of edges incident to a specific vertex of $D$ is defined by the order of the shortest paths in $G$ that they represent.  Due to space constraint we give the complete details of this process which takes total time of $O(n \log^3 n)$ in Appendix \ref{apx:C}. Then, we ignore all non-tight edges in $D$ and find $\ell(x)$ by finding the leftmost path from $u'$ to $v'$. This takes time proportional to the number of vertices in $D$, which is $O(\sqrt{n})$.

We perform the computation of parts (iv) and (v) for every $B$ and $B'$, and $x \in X$. For every $x$ this computation takes $O(\sqrt{n} \log^2 n)$ time, and we repeat it $O(n^2 / r^{3/2})$ times. The total time required for constructing the data structure is $O((n^{5/2} / r^{3/2}) \log^2 n + n\sqrt{r})$ (note that the term $n \log^3 n$ is dominated by this bound). We note that it is also possible to perform the computation in $O((n^3/r^2) \log n)$ time using the algorithm of Klein \cite{K05} for every pair of pieces, however this does not improve the time bound for our range of $r$.

\subsection{The Query Algorithm} \label{sec:Cq}

Let $u, v$ be a query pair. We use the data structure of this section to find $d_G(u, v)$ in $O(\sqrt{r})$ time. Again, let $B$, $B'$ be the pieces that contain $u$, $v$ respectively.

If $B = B'$ then we can answer the query in $O(\sqrt{r})$ time in same way as in the previous data structure, which uses the query algorithm of \cite{C06}. Either the shortest path from $u$ to $v$ is inside $B$ or the shortest path contains some vertex $b \in \partial B$. In other words, $d_G(u, v) = \min \lbrace d_B(u, v), \min_{b \in \partial B} \{d_G(u, b) + d_G(b, v)\} \rbrace$. We retrieve the distance $d_B(u, v)$ from part (ii) of the data structure in $O(\sqrt{r})$ time. For each $b \in \partial B$ we retrieve $d_G(u, b) + d_G(b, v)$ from part (i) in $O(1)$ time. Since $|\partial B| = O(\sqrt{r})$, it takes $O(\sqrt{r})$ time to go over all vertices of $\partial B$ and find $b \in \partial B$ that minimizes $d_G(u, b) + d_G(b, v)$. Then, $d_G(u, v)$ is the minimum between $d_B(u, v)$ and $d_G(u, b) + d_G(b, v)$.

For the case where $u$ and $v$ are in different pieces, let $X$ and $Y$ be as before. Fix $x_1$ to be an arbitrary vertex of $X$ and let $y_{|Y|} = \ell(x_1)$. Let the matrices $N$, $N'$, $N''$ be as before. We compute an entry $N_{ij} =  d_G(u, x_i) + d_J(x_i, y_j) + d_G(y_j, v)$ in $O(1)$ time using parts (i) and (iv) of the data structure. We determine whether an entry of $N$ is in $N'$ or in $N''$ in $O(1)$ time using part (v). Therefore, we can use the SMAWK algorithm \cite{AKM+87} as in \cite{AK90} to find the minimum value in $N$ in $O(\sqrt{r})$ time. This value is the requested distance.

We conclude that
for a planar graph with $n$ vertices and any $r \in (0, n)$, we can construct in $O((n^{5/2} / r^{3/2}) \log^2 n + n \sqrt{r})$ time a data structure of size $O(n^2 / r + n\sqrt{r})$ that computes the distance between any two vertices in $O(\sqrt{r})$ time.
As in Sect.~\ref{sec:B}, we set $r = n^2 / S$ and get:
\begin{theorem}
  For a planar graph with $n$ vertices and $S \in [n^{4/3},n^2]$, we can construct in $O((S^{3/2} / \sqrt{n}) \log^2 n)$ time a data structure of size $O(S)$ that computes the distance between any two vertices in $O(n / \sqrt{S})$ time.
\end{theorem}

\appendix

\section{The Query Algorithm of Section \ref{sec:B} when the Vertices Are in Two Different Pieces}\label{apx:B}

Consider the data structure of Sect.~\ref{sec:B}. Let $u, v$ be a query pair, such that $u$ is in a piece $B$ and $v$ is in another piece $B'$. The query algorithm that we describe here is similar to the one of \cite[(\S5)]{D96}, and the complete details are given there.\footnote{There is a small difference in our description of the algorithm; instead of Step 1 of \cite[(\S5)]{D96} we use part (i) of the data structure which \cite[(\S5)]{D96} does not have.} Let $H$ be the hole of $B$ that contains $v$ and let $H'$ the hole of $B'$ that contains $u$. Denote $X = \partial B[H]$ and $Y = \partial B'[H']$. We assume without loss of generality that $H \cap H'$ contains the infinite face.

A shortest path from $u$ to $v$ contains some vertex $x \in X$ and some vertex $y \in Y$ (it is possible that $x = y$). We may assume that there is no internal vertex of $B$ between $x$ and $y$ in the shortest path (since otherwise we can replace $x$ with another vertex of $X$). Therefore, $d_G(u, v) = \min_{x \in X, y \in Y} \{d_G(u, x) + d_H(x, y) + d_G(y, v)\}$. Next we show how to find $\min_{y \in Y}\{d_H(x, y) + d_G(y, v)\}$ for every $x \in X$ in $O(\sqrt{r} \log r)$ time.

For a fixed vertex $x_i \in X$ it is easy to find $y_i \in Y$ that minimizes $d_H(x_i, y_i) + d_G(y_i, v)$ in $O(\sqrt{r})$ time (the same member of $Y$ may be $y_i$ for different members of $X$), by going over all vertices of $Y$ and using parts (iii) and (i) of the data structure. Let $y_1$ and $y_2$ be the vertices that minimizes $d_H(x_i, y_i) + d_G(y_i, v)$ for $x_1$ and $x_2$. There is a shortest path from $x_1$ to $v$ that contains $y_1$, and similarly a shortest path from $x_2$ to $v$ that contains $y_2$. Let $x_3$ be a vertex between $x_1$ and $x_2$ in the clockwise order of $X$ starting at $x_1$. There is a vertex $y_3 \in Y$ that minimizes $d_H(x_3, y_3) + d_G(y_3, v)$ located between $y_1$ and $y_2$ in the counterclockwise order of vertices of $Y$ starting at $y_1$. Since otherwise, every shortest path from $x_3$ to $v$ must cross the shortest path from $x_1$ or from $x_2$ to $v$ that contains $y_1$ or $y_2$, respectively. Assume without loss of generality that the shortest path from $x_3$ to $v$ crosses the shortest path from $x_1$ to $v$, and let $w$ be the vertex in which the two shortest paths meet. Then, if we replace the suffix of the shortest path from $x_1$ that begins at $w$ with the suffix of the shortest path from $x_3$ we get a shorter path, this is a contradiction (see Fig.~\ref{fig:x3y3}). This gives the following algorithm for finding $y_i \in Y$ for every $x_i \in X$.

\begin{figure}[!t]
	\centering
	\includegraphics[scale=0.5]{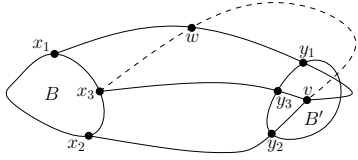}
	\caption{Finding $y_3$. For $x_3$ between $x_1$ and $x_2$, the vertex $y_3$ is between $y_1$ and $y_2$, since otherwise we get that every path from $x_3$ to $v$ (\emph{dashed}) crosses either the shortest path from $x_1$ to $v$ or from $x_2$ to $v$ at a vertex $w$.}
	\label{fig:x3y3}
\end{figure}

Let $x_1, x_2$ be two arbitrary vertices of $X$. Find $y_1$ and $y_2$ for $x_1$ and $x_2$ by going over all vertices of $Y$. Let $x_3$ be the middle vertex between $x_1$ and $x_2$ in the clockwise order of $X$ starting at $x_1$. Find $y_3$ by going over all vertices of $Y$ between $y_1$ and $y_2$ in the counterclockwise order of vertices of $Y$ starting at $y_1$. Continue recursively for the vertices of $X$ between $x_1$ and $x_3$ and the vertices of $Y$ between $y_1$ and $y_3$, and also for the vertices of $X$ between $x_3$ and $x_2$ and the vertices of $Y$ between $y_3$ and $y_2$. Similarly, find $y_i$ for every $x_i$ between $x_1$ and $x_2$ in the counterclockwise order of $X$ starting at $x_1$.

We conclude that we can find $\min_{y \in Y}\{d_H(x, y) + d_G(y, v)\}$ for every $x \in X$ in $O(\sqrt{r} \log r)$ time. Now, we go over all vertices of $X$, and using part (i) of the data structure we find $\min_{x \in X, y \in Y} \{d_G(u, x) + d_H(x, y) + d_G(y, v)\}$ in $O(\sqrt{r})$ time. The total query time is $O(\sqrt{r} \log r)$.

\section{A Cyclic Order for Edges Incident to a Vertex of $D$}\label{apx:C}

In this appendix we define a cyclic order on the edges incident to a specific vertex in the graph $D$, which is a subgraph of the dense distance graph. We use this order in the preprocessing algorithm of Sect.~\ref{sec:Cp}, in order to find $\ell(x)$ for a boundary vertex $x$. We define the order of the edges such that the leftmost shortest paths from $x$ to $Y$ in $G$, and in $D$, both end at the same vertex of $Y$ ($x, Y, \ell(x)$ and $D$ are defined in Sect.~\ref{sec:C}). A vertex $w$ of $D$ is a boundary vertex of more than one piece, however the order between two edges in two different pieces is clear from the embedding of $G$ (the pieces of $D$ are pairwise edge disjoint). Therefore, here we define the left-to-right order of the edges inside each piece. The left-to-right order of the edges, is in fact a left-to-right order of the boundary vertices, because the edges of a piece in the dense distance graph connect a vertex on the boundary of the piece to all other vertices on the boundary. We define the left-to-right order from $w$ to the other boundary vertices of the piece according to the left-to-right order of the leftmost shortest paths from $w$ to the other vertices. This order allows us to find $\ell(x)$ as required. The order that we define does not depend on the specific graph $D$, so we perform the procedure described here only once for every boundary vertex of every piece.

Let $w$ be a boundary vertex of a piece $C$. When we compute the distances from $w$ to the other vertices of $\partial C$ for the dense distance graph, we use the algorithm of Klein \cite{K05}, which maintains a \emph{dynamic tree} \cite{ST83} that contains the rightmost shortest path from $w$ to every vertex of $C$. Since we are interested in leftmost shortest paths we use a symmetric version of \cite{K05}, by replacing the roles of left and right. Denote this leftmost shortest path tree rooted at $w$ by $T$.

Let $z$ and $z'$ be two vertices of $\partial C$ different from $w$. We show how to decide in $O(\log |C|)$ time which of the two vertices is to the left of the other, with respect to $w$. Let $t$ be the nearest common ancestor of $z$ and $z'$ in $T$. We can find $t$ and the two edges that lead from it to $z$ and to $z'$ in $O(\log |C|)$ time from the dynamic tree \cite{ST83}. First assume that $t \neq z, z'$. Consider the following three edges incident to $t$ in $T$ -- the edge that connects $t$ to its parent (if $t$ = $w$ then we add a dummy edge inside the hole that $w$ lies on its boundary for this purpose), the edge that leads from $t$ to $z$, and the edge that leads from $t$ to $z'$. The order of these edges around $t$ determine the order between $z$ and $z'$ (see Fig.~\ref{fig:compare}(a)). Now assume without loss of generality that $t = z$. The vertex $z$ lies on the boundary of some hole of $C$, denote this hole by $H$. There are two edges incident to $z$ on the boundary of $H$. We can find the two edges when we find the piece $C$. Consider the edge that connects $z$ to its parent in $T$, the edge that leads from $z$ to $z'$, and the place of $H$ among the edges incident to $z$. If in the clockwise order of edges around $z$ starting at the edge that connects $z$ to it parent, the edge that leads to $z'$ is before $H$, then $z'$ is to the left of $z$, otherwise $z$ is to the left of $z'$ (see Fig.~\ref{fig:compare}(b)).

\begin{figure}[t]
	\centering
	\subfigure[$t$ is the nearest common ancestor of $z$ and $z'$.]{\hspace{1.5cm}\includegraphics[scale=0.45]{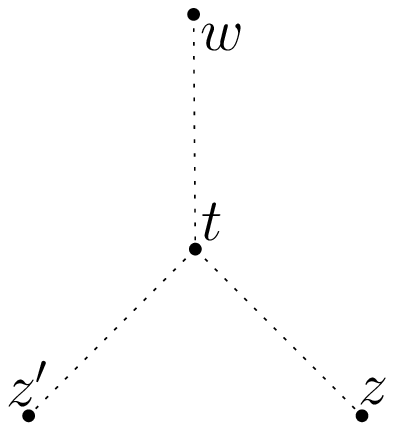}\hspace{1.5cm}} \qquad
	\subfigure[$z$ is an ancestor of $z'$ and lies on the boundary of $H$.]{\hspace{1.5cm}\includegraphics[scale=0.45]{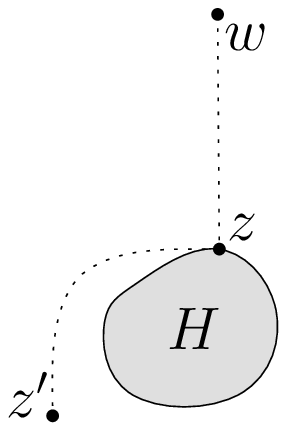}\hspace{1.5cm}} \\
	\caption{Two cases where $z$ is to the left of $z'$. The \emph{dotted} lines represent paths in $T$.}
	\label{fig:compare}
\end{figure}

Since we compare two vertices in $O(\log |C|)$ time, we can use comparison sort to sort the $O(\sqrt{|C|})$ vertices of $\partial C$ around the vertex $w$ from left to right in $O(\sqrt{|C|} \log^2 |C|)$ time. We repeat the process for each vertex of $\partial C$ in a total of $O(|C| \log^2 |C|)$ time. For the pieces of a single layer of the recursive decomposition, the total time is $O(n \log ^2 n)$. And for all the pieces of the dense distance graph the process takes $O(n \log^3 n)$ time.

\end{document}